\documentclass[a4paper,11pt]{article}

\usepackage{amsfonts}
\usepackage{amsmath}
\usepackage{amsthm}
\usepackage{amssymb}
\usepackage{mathtools}
\usepackage[hyphens]{url}
\usepackage{hyperref}
\usepackage{graphicx}
\usepackage[ruled]{algorithm}
\usepackage[noend]{algpseudocode}
\usepackage{framed}

\usepackage[left=2.5cm,top=3cm,right=2.5cm,bottom=3cm,bindingoffset=0.5cm]{geometry}

\newtheorem{theorem}{Theorem}

\newtheorem{lemma}{Lemma}

\newtheorem{remark}{Remark}

\begin{document}
\title{TuringMobile: A Turing Machine of Oblivious Mobile Robots with Limited Visibility and its Applications}
\author{Giuseppe A.~Di Luna\footnotemark[1] \and Paola Flocchini\footnotemark[2] \and Nicola Santoro\footnotemark[3] \and Giovanni Viglietta\footnotemark[4]}
\date{}
\thispagestyle{empty}
\renewcommand{\thefootnote}{\fnsymbol{footnote}}
\footnotetext[1]{Aix-Marseille University and LiS Laboratory, France. E-mail: \texttt{giuseppe.diluna@lif.univ-mrs.fr}.}
\footnotetext[2]{University of Ottawa, Canada. E-mail: \texttt{paola.flocchini@uottawa.ca}.}
\footnotetext[3]{Carleton University, Canada. E-mail: \texttt{santoro@scs.carleton.ca}.}
\footnotetext[4]{JAIST, Japan. E-mail: \texttt{johnny@jaist.ac.jp}.}
\renewcommand{\thefootnote}{\arabic{footnote}}
\setcounter{footnote}{0}

\maketitle

\begin{abstract}
In this paper we investigate the computational power of a set of mobile robots with limited visibility.
At each iteration, a robot takes a snapshot of its surroundings, uses the snapshot to compute a destination point, and it moves toward its destination. Each robot is punctiform and memoryless, it operates in $\mathbb{R}^m$, it has a local reference system independent of the other robots' ones, and is activated asynchronously by an adversarial scheduler. Moreover, the robots are non-rigid, in that they may be stopped by the scheduler at each move before reaching their destination (but are guaranteed to travel at least a fixed unknown distance before being stopped).

We show that despite these strong limitations, it is possible to arrange $3m+3k$ of these weak entities in $\mathbb{R}^m$ to simulate the behavior of a stronger robot that is rigid (i.e., it always reaches its destination) and is endowed with $k$ registers of persistent memory, each of which can store a real number. We call this arrangement a {\em TuringMobile}. In its simplest form, a TuringMobile consisting of only three robots can travel in the plane and store and update a single real number. We also prove that this task is impossible with fewer than three robots.

Among the applications of the TuringMobile, we focused on Near-Gathering (all robots have to gather in a small-enough disk) and Pattern Formation (of which Gathering is a special case) with limited visibility. Interestingly, our investigation implies that both problems are solvable in Euclidean spaces of any dimension, even if the visibility graph of the robots is initially disconnected, provided that a small amount of these robots are arranged to form a TuringMobile. In the special case of the plane, a basic TuringMobile of only three robots is sufficient.
\end{abstract}

\section{Introduction}
\paragraph*{Framework and background}

The investigations of systems of autonomous mobile robots  have long moved
outside the boundaries of the engineering, control, and AI communities.
Indeed,  the computational and complexity issues arising  in such systems
are  important  research  topics within  theoretical computer science,
especially in distributed computing.
In these theoretical investigations, the robots are usually viewed as
punctiform computational entities  that live in a metric space, typically $\mathbb
R^2$ or $\mathbb R^3$, in which they can move.
 Each  robot operates in ``Look-Compute-Move'' (LCM) cycles: it observes
its surroundings, it computes a destination within the space  based on what it sees, and it moves toward the destination.
 The only means of  interaction between
 robots are observations and  movements: that is,  communication is {\em
stigmergic}.
 The robots, identical and outwardly indistinguishable,   are {\em
oblivious}:  when starting a new
 cycle, a robot  has no memory of its activities (observations,
computations, and moves) from  previous cycles (``every time is  the
first time'').
 In other words, the  robots have no persistent memory; for this reason,
they are sometime said to be memoryless.
 Clearly obliviousness is a desirable property as it ensures a  degree of
self-stabilization and fault-tolerance into the system and its computations.
Equally clear is that being memoryless severely constrains the
computational capabilities of the robots.

There have been  intensive research efforts on  the computational issues arising
with such robots, and
 an extensive  literature  has been produced  in
particular in regard to  the important class of
 {\em Pattern Formation} problems~\cite{xpatternnavarra,xxFlPSV17,xxFlPSW08,xxFuYKY15,xxSuzY99,xYamS10} as well as for {\em
Gathering}~\cite{xfatGathering,xAgP06,xxAnOSY99,xgatehringnavarra,xCiFPS12,xCohP05,xseb1,xseb2,xgatheringdefago,xxFlPSW05,xPaPV15};  and {\em Scattering}~\cite{xbramas,xscatteringconn};  see also~\cite{xflockingcanepa,xxDaFSY15,xYuUKY17}.
The goal  of the research has been to  understand the minimal
assumptions needed for a team (or swarm) of such robots to solve
a given problem, and to identify the impact that specific factors have on
feasibility and hence computability.

The most important factor is the power of  the adversarial scheduler that
decides when each activity of each robot starts and when it ends.
The main   adversaries (or ``environments'') considered  in the
literature are:    {\em synchronous},  in which
the computation cycles of all active robots are synchronized,
and at each cycle either all (in the fully synchronous case) or a subset (in the semi-synchronous case) of the robots are activated,
and {\em asynchronous}, where computation cycles are not
synchronized, each activity  can take a different and unpredictable amount of time,
and each robot can be independently activated at each time instant.
An important factor is whether a robot moving toward a computed
destination is guaranteed to reach it (i.e., it is a {\em rigid} robot), or it can be
stopped on the way (i.e., it is a {\em non-rigid} robot) at a point  decided by an
 adversary. In all the above cases, the power of the adversaries is limited by some
basic fairness assumption.
All the existing  investigations have concentrated on the study of (a-)synchrony,
several on the impact of rigidity, some on other relevant factors such as agreement on local coordinate systems or on their orientation, etc.; for a review, see~\cite{xxFlPS12}.

From a computational point of view, there is another crucial factor: the visibility range of the robots, that is, how much of
the surrounding space they are able to observe in a Look operation. In this regard,  two basic settings are considered:
 {\em unlimited visibility}, where the robots can see the entire space
(and thus all other robots), and  {\em limited  visibility}, when
the robots have a fixed visibility radius.
While  the investigations on  (a-)synchrony and rigidity have concentrated on all aspects of those assumptions,
this is not the case with respect to visibility.
In fact, almost all research has assumed unlimited visibility;
 few exceptions
are the  algorithms for Convergence~\cite{xxAnOSY99}, Gathering~\cite{xxxDeKLHPW11,xxDeKLHPW11,xxFlPSW05}, and
Near-Gathering \cite{xPaPV15} when  the visibility range of the robot is limited.
The unlimited visibility assumption clearly greatly
 simplifies the computational universe under investigation;  at the same time,
 it neglects the more general and realistic one, which is still largely unknown.

Let us also stress that, in the existing literature, all results  on
oblivious robots are for  $\mathbb R^1$ and $\mathbb R^2$; the  only exception is the recent result on  plane formation in  $\mathbb R^3$
by semi-synchronous  rigid robots
with unlimited visibility~\cite{xYuUKY17}. No results exist for robots in  higher dimensions.

\paragraph*{Contributions}
In this paper we contribute several constructive insights on the
computational universe of  oblivious robots with limited visibility,
especially asynchronous non-rigid ones, in any dimension.

The first and main contribution  is
a technique to  construct a ``moving Turing Machine'' made solely of  asynchronous
oblivious non-rigid robots in $\mathbb R^m$ with limited visibility, for
any $m\geq 2$.
More precisely, we show how to arrange  $3m+3k$ identical non-rigid
oblivious robots in $\mathbb R^m$ with a visibility radius of
$V+\varepsilon$ (for any $\varepsilon>0$) and how  to program them so that
they can collectively behave as a single rigid robot in
$\mathbb R^m$ with $k$ persistent registers and visibility radius $V$
would.
This team of identical robots is informally called a \emph{TuringMobile}.

We obtain this result by using as fundamental construction a basic component,
which is able to move in $\mathbb{R}^2$ while storing and updating a single real
number. Interestingly, we show that $3$ agents are necessary and
sufficient to build such a machine. The TuringMobile will then be
built by arranging multiple copies of this basic component side by side.

We stress that robots forming the TuringMobile are \emph{asynchronous},
that is, the scheduler makes them move at independent arbitrary speeds, and
each robot takes the next snapshot an arbitrary amount of time after
terminating each move; furthermore, they are \emph{anonymous}, in that they
are indistinguishable from each other, and they all execute the same
program to compute their destination points. Notably, this program only performs arithmetic operations, square roots, and comparisons (hence no transcendental function has to be computed by the robots).

 A TuringMobile is  a powerful construct that, once deployed in a swarm of robots, can  act as
 a  rigid leader  with persistent memory, allowing the swarm to overcome  many
 handicaps imposed by obliviousness, limited visibility, and asynchrony.
 As examples we present a variety of applications in  $\mathbb{R}^m$, with $m\geq 2$.

First of all we show how a TuringMobile can explore and search the space.
 We then show how it can be employed
 to solve
  the  long-standing open problem of (Near-)Gathering with limited
visibility  in spite of an asynchronous
non-rigid scheduler  and disagreement on the axes,
a problem still open without a TuringMobile.
 Interestingly, the presence of the TuringMobile allows Gathering to be done
 even if the initial visibility graph is disconnected (this does not change the fact that there are cases in which Gathering is impossible, as remarked in \cite{xxAnOSY99,xxFlPSW05}). Finally we show how the  arbitrary Pattern Formation problem can be solved
under the same conditions (asynchrony, limited visibility, possibly disconnected
visibility graph, etc.).

There is a limitation to the use of a TuringMobile when deployed in a swarm of robots.
Namely, the TuringMobile must be always recognizable
(e.g., by its  unique shape) so that  other robots cannot interfere
 by  moving too close to the machine, disrupting
its structure.

The paper is organized as follows: In Section \ref{s2} we give formal
definitions, introducing mobile robots with or without memory as \emph{oracle semi-oblivious real RAMs}. In Section \ref{s3} we
illustrate our implementation of the TuringMobile. The correctness of the
proposed construction is proved in Section \ref{s4}. In Section \ref{s5} we
show how to apply the TuringMobile to solve fundamental problems. In Section~\ref{s6} we conclude with some extra remarks and open problems.

\section{Definitions and Preliminaries}\label{s2}

\subsection{Oracle Semi-Oblivious Real RAMs}
\paragraph*{Real random-access machines}
A \emph{real RAM}, as defined by Shamos~\cite{xram2,xram1}, is a random-access machine~\cite{xram0} that can operate on real numbers. That is, instead of just manipulating and storing integers, it can handle arbitrary real numbers and do infinite-precision operations on them. It has a finite set of internal \emph{registers} and an infinite ordered sequence of \emph{memory cells}; each register and each memory cell can hold a single real number, which the machine can modify by executing its program.\footnote{Nonetheless, the constant operands in a real RAM's program cannot be arbitrary real numbers, but have to be integers.}

A real RAM's instruction set contains at least the four arithmetic operations, but it may also contain $k$-th roots, trigonometric functions, exponentials, logarithms, and other analytic functions, depending on the application. The machine can also compare two real numbers and branch depending on which one is larger.

The initial contents of the memory cells are the \emph{input} of the machine (we stipulate that only finitely many of them contain non-zero values), and their contents when the machine halts are its \emph{output}. So, each program of a real RAM can be viewed as a partial function mapping tuples of reals into tuples of reals.

\begin{remark}
The real RAMs can at least compute the Turing-computable partial functions over the integers. Indeed, it is well known that all these functions can be computed by traditional RAMs whose programs only contain integer additions, subtractions, and comparisons. It is obvious that a real RAM running such a program on an integer-valued input behaves exactly as a traditional RAM, and therefore computes the same partial function.
\end{remark}

\paragraph*{Oracles and semi-obliviousness}
We introduce the \emph{oracle semi-oblivious real RAM}, which is a real RAM with an additional ``ASK'' instruction. Whenever this instruction is executed, the contents of all the memory cells are replaced with new values, which are a function of the numbers stored in the registers.

In other words, the machine can query an external oracle by putting a question in its $k$ registers in the form of $k$ real numbers. The oracle then reads the question and writes the answer in the machine's memory cells, erasing all pre-existing data. The term ``semi-oblivious'' comes from the fact that, every time the machine invokes the oracle, it ``forgets'' everything it knows, except for the contents of the registers, which are preserved.\footnote{Observe that, in general, the machine cannot salvage its memory by encoding its contents in the registers: since its instruction set has only analytic functions, it cannot injectively map a tuple of arbitrary real numbers into a single real number.}

\begin{figure}[ht]
\begin{center}
\includegraphics[width=\linewidth]{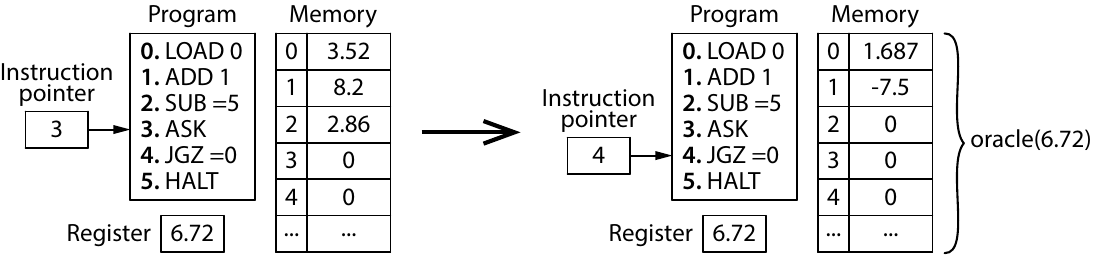}
\end{center}
\caption{An oracle semi-oblivious real RAM with a single register before and after executing an ``ASK'' instruction. The entire memory is overwritten by the oracle based on the number read from the register, which remains unaltered}
\label{f:ram}
\end{figure}

In spite of their semi-obliviousness, these real RAMs with oracles are at least as powerful as Turing Machines with oracles.

\begin{theorem}
Given an oracle Turing Machine, there is an oracle semi-oblivious real RAM with one register that computes the same partial function.
\end{theorem}
\begin{proof}
Following Rogers~\cite{xrogers}, we define an oracle Turing Machine as a Turing Machine with an additional read-only tape containing the answers to all possible oracle queries. The $i$th cell of the oracle tape contains a symbol that is read by the machine whenever the head of the oracle tape is in position $i$.

Given such a machine $M$, we construct an oracle semi-oblivious real RAM with one register $M'$ that ``simulates'' $M$ step by step. As already observed, a real RAM can compute any Turing-computable partial function, and $M'$ behaves as a real RAM as long as it does not invoke its oracle. So, $M'$ can encode and decode the entire state of $M$, including the contents of its non-oracle tapes and the positions of its heads on the tapes, as a single integer: indeed, the functions that encode and decode a Turing Machine's state are themselves Turing-computable.

To simulate one step of $M$, $M'$ encodes the current state of $M$ in its register and executes an ``ASK'' instruction. The oracle of $M'$ reads the register, decodes the state of $M$, fetches the position of the head on the oracle tape, and answers with the symbol $s$ that $M$ would read on its oracle tape at that position. Next, $M'$ finds $s$ in the first cell of its own memory. So, $M'$ decodes the contents of the register to retrieve the state of $M$, and uses it along with $s$ to compute the next state of $M$.
\end{proof}

\subsection{Mobile Robots as Real RAMs}\label{s2:2}
\paragraph*{Mobile robots}
Our oracle semi-oblivious real RAM model can be reinterpreted in the realm of \emph{mobile robots}. A mobile robot is a computational entity, modeled as a geometric point, that lives in a metric space, typically $\mathbb R^2$ or $\mathbb R^3$. It can observe its surroundings and move within the space based on what it sees. The same space may be populated by several mobile robots, each with its local coordinate system, and static objects.

To compute its next destination point, a mobile robot executes a real RAM program with input a representation of its local view of the space. After moving, its entire memory is erased, but the content of its $k$ registers is preserved. Then it makes a new observation; from the observation data and the contents of the registers, it computes another destination point, and so on. If $k=0$, the mobile robot is said to be \emph{oblivious}. Note that robots have no notion of time or absolute positions.

The actual movement of a mobile robot is controlled by an external \emph{scheduler}. The scheduler decides how fast the robot moves toward its destination point, and it may even interrupt its movement before the destination point is reached. If the movement is interrupted midway, the robot makes the next observation from there and computes a new destination point as usual. The robot is not notified that an interruption has occurred, but it may be able to infer it from its next observation and the contents of its registers. For fairness, the scheduler is only allowed to interrupt a robot after it has covered a distance of at least $\delta$ in the current movement, where $\delta$ is a positive constant unknown to the robots. This guarantees, for example, that if a robot keeps computing the same destination point, it will reach it in a finite number of iterations. If $\delta=\infty$, the robot always reaches its destination, and is said to be \emph{rigid}.

\paragraph*{Mobile robots, revisited}
A mobile robot in $\mathbb R^m$ with $k$ registers can be modeled as an oracle semi-oblivious real RAM with $2m+k+1$ registers, as follows.

\begin{itemize}
\item $m$ \emph{position registers} hold the absolute coordinates of the robot in $\mathbb R^m$.
\item $m$ \emph{destination registers} hold the destination point of the robot, expressed in its local coordinate system.
\item $1$ \emph{timestamp register} contains the time of the robot's last observation.
\item $k$ \emph{true registers} correspond to the registers of the robot.
\end{itemize}

As the RAM's execution starts, it ignores its input, erases all its registers, and executes an ``ASK'' instruction. The oracle then fills the RAM's memory with the robot's initial position $p$, the time $t$ of its first observation, and a representation of the geometric entities and objects surrounding the robot, as seen from $p$ at time $t$.

The RAM first copies $p$ and $t$ in its position registers and timestamp register, respectively. Then it executes the program of the mobile robot, using its true registers as the robot's registers and adding $m+1$ to all memory addresses. This effectively makes the RAM ignore the values of $p$ and $t$, which indeed are not supposed to be known to the mobile robot.

When the robot's program terminates, the RAM's memory contains the output, which is the next destination point $p'$, expressed in the robot's coordinate system. The RAM copies $p'$ into its destination registers, and the execution jumps back to the initial ``ASK'' instruction.

Now the oracle reads $p$, $p'$, and $t$ from the RAM's registers (it ignores the true registers), converts $p'$ in absolute coordinates (knowing $p$ and the orientation of the local coordinate system of the robot) and replies with a new position $p''$, a timestamp $t'>t$, and observation data representing a snapshot taken from $p''$ at time $t'$. To comply with the mobile robot model, $p''$ must be on the segment $pp'$, such that either $p''=p'$ or $\overline{pp''}\geq\delta$. The execution then proceeds in the same fashion, indefinitely.

Note that in this setting the oracle represents the scheduler. The presence of a timestamp in the query allows the oracle to model dynamic environments in which several independent robots may be moving concurrently (without a timestamp, two observations from the same point of view would always be identical). Also note that in this formulation there are no actual robots moving through an environment in time, but only RAMs which query an oracle, which in turn provides a ``virtual'' environment and timeline by writing information in their memory.

\paragraph*{Snapshots and limited visibility}
In the mobile robot model we consider in this paper, an observation is simply an instantaneous \emph{snapshot} of the environment taken from the robot's position. In turn, each entity and object that the robot can see is modeled as a dimensionless point in $\mathbb R^m$. A mobile robot has a positive \emph{visibility radius} $V$: it can see a point in $\mathbb R^m$ if and only if it is located at distance at most $V$ from its current position. If $V=\infty$, the robot is said to have \emph{unlimited visibility}.

As we hinted at earlier in this section, a mobile robot has its own local reference system in which all the coordinates of the objects in its snapshots are expressed. The origin of a robot's local coordinate system always coincides with the robot's position (hence it follows the robot as it moves), and its orientation and handedness are decided by the scheduler (and remain fixed). Different mobile robots may have coordinate systems with a different orientation or handedness. (However, when two robots have the same visibility radius, they also implicitly have the same unit of distance.)

So, a snapshot is just a (finite) list of points, each of which is an $m$-tuple of real numbers.

\paragraph*{Simulating memory and rigidity}
The main contribution of this paper, loosely speaking, is a technique to turn non-rigid oblivious robots into rigid robots with persistent memory, under certain conditions. More precisely, if $3m+3k$ identical non-rigid oblivious robots in $\mathbb R^m$ with a visibility radius of $V+\varepsilon$ (for any $\varepsilon>0$) are arranged in a specific pattern and execute a specific algorithm, they can collectively act in the same way as a single rigid robot in $\mathbb R^m$ with $k>0$ persistent registers and visibility radius $V$ would. This team of identical robots is informally called a \emph{TuringMobile}.

We stress that the robots of a TuringMobile are \emph{asynchronous}, that is, the scheduler makes them move at independent arbitrary speeds, and each robot takes the next snapshot an arbitrary amount of time after terminating each move. The robots are also \emph{anonymous}, in that they are indistinguishable from each other, and they all execute the same program.

Although our technique is fairly general and has a plethora of concrete applications (some are discussed in Section~\ref{s5}), a ``perfect simulation'' is achieved only under additional conditions on the scheduler or on the environment. These conditions will be discussed toward the end of Section~\ref{s3:2}.

\section{Implementing the TuringMobile}\label{s3}
\subsection{Basic Implementation}\label{s3:1}
We will first describe how to construct a basic version of the TuringMobile with just three oblivious non-rigid robots in $\mathbb R^2$. This TuringMobile can remember a single real number and rigidly move in the plane by fixed-length steps: its layout is sketched in Figure~\ref{f:machine}. In Section~\ref{s3:2}, we will show how to combine several copies of this basic machine to obtain a full-fledged TuringMobile.

\begin{figure}[ht]
\begin{center}
\includegraphics[width=\linewidth]{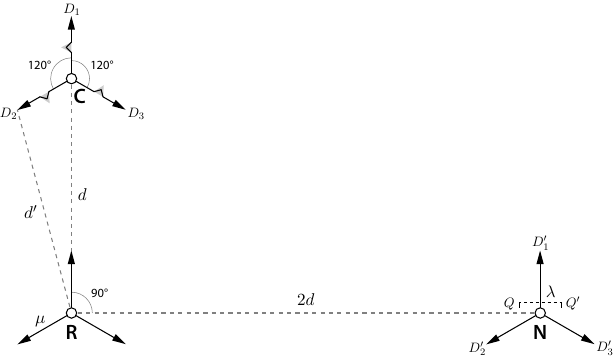}
\end{center}
\caption{Basic TuringMobile at rest, not drawn to scale ($\mu$ and $\lambda$ should be smaller)}
\label{f:machine}
\end{figure}

\paragraph*{Position at rest}
The elements of the basic TuringMobile are three robots: a \emph{Commander} robot, a \emph{Number} robot, and a \emph{Reference} robot, located in $C$, $N$, and $R$, respectively. These robots have the same visibility radius of $V+\varepsilon$, where $\varepsilon \ll V$, and there is always a disk of radius $\varepsilon$ containing all three of them. When the machine is ``at rest'', $\angle NRC$ is a right angle, the distance between $C$ and $R$ is some fixed value $d\ll\varepsilon$, and the distance between $R$ and $N$ is approximately $2d$. More precisely, $N$ lies on a segment $QQ'$ of length $\lambda$, where $\lambda\ll d$ is some fixed value, such that $Q$ has distance $2d-\lambda/2$ from $R$ and $Q'$ has distance $2d+\lambda/2$ from $R$.

\paragraph*{Representing numbers}
The distance between the Reference robot and the Number robot when the TuringMobile is at rest is a representation of the real number $r$ that the machine is currently storing. There are several possible ways of defining such a code: an easy one is to encode the number $r$ as $\overline{RN}=\alpha(r)=2d+\arctan(r)\cdot \lambda/\pi$ and to decode it as $r=\tan\left( (\overline{RN}-2d)\cdot\pi/\lambda \right)$. A different method that does not use transcendental functions is discussed in Section~\ref{s6}.

\paragraph*{Movement directions}
The Commander's role is to decide in which direction the machine should move next, and to initiate the movement. When the machine is at rest, the Commander may choose among three possible \emph{final destinations}, labeled $D_1$, $D_2$, and $D_3$ in Figure~\ref{f:machine}. The segments $CD_1$, $CD_2$, and $CD_3$ all have the same length $\mu$, with $\lambda\ll\mu\ll d$, and form angles of $120^\circ$ with one another, in such a way that $D_1$ is collinear with $R$ and $C$.

Around the center of each segment $CD_i$ there is a \emph{midway triangle} $\tau_i$, drawn in gray in Figure~\ref{f:machine}. This is an isosceles triangle of height $\lambda$ whose base lies on $CD_i$ and has length $\lambda$ as well. When the Commander decides that its final destination is $D_i$, it moves along the segment $CD_i$, but it takes a small detour in the midway triangle $\tau_i$, as we will explain shortly.

\paragraph*{Structure of the algorithm}
Algorithm~\ref{a:1} is the program that each element of the basic TuringMobile executes every time it computes its next destination point.

\begin{algorithm}
\caption{Basic TuringMobile in $\mathbb R^2$\label{a:1}}
\begin{algorithmic}[1]
\State Identify \emph{Commander}, \emph{Number}, \emph{Reference} (located in $C$, $N$, $R$, respectively)\label{r:1}
\If{I am \emph{Commander}}\label{r:2}
	\State Compute \emph{Virtual Commander} $C'$ (based on $R$ and $N$) and points $A_i$, $S_i$, $S'_i$, $B_i$, $D_i$\label{r:3}
	\If{I am in $C'$}\ Choose final destination $D_i$ and move to $A_i$\label{r:4}
	\ElsIf{$\exists i\in\{1,2,3\}$ s.t.\ I am on segment $C'A_i$ but not in $A_i$}\ Move to $A_i$\label{r:5}
	\ElsIf{$\exists i\in\{1,2,3\}$ s.t.\ I am in $A_i$}\label{r:5:1}
		\State Move to point $P$ on segment $S_iS'_i$ such that $\overline{PS_i} = f(\overline{NQ})$\label{r:5:2}
	\ElsIf{$\exists i\in\{1,2,3\}$ s.t.\ I am in triangle $A_iS_iS'_i$ but not on segment $S_iS'_i$}\label{r:6}
		\State Move to the intersection of segment $S_iS'_i$ with the extension of line $A_iC$\label{r:7}
	\ElsIf{$\exists i\in\{1,2,3\}$ s.t.\ I am on $S_iS'_i$ \textbf{and} $\overline{NQ}=\overline{CS_i}$}\ Move to $B_i$\label{r:8}
	\ElsIf{$\exists i\in\{1,2,3\}$ s.t.\ I am in triangle $B_iS_iS'_i$ but not in $B_i$}\ Move to $B_i$\label{r:9}
	\ElsIf{$\exists i\in\{1,2,3\}$ s.t.\ I am on segment $B_iD_i$ but not in $D_i$}\ Move to $D_i$\label{r:10}
	\EndIf
\ElsIf{I am \emph{Number}}\label{r:11}
	\If{$\overline{CR}=d+\mu$ \textbf{or} $\overline{CR}=d'$}\label{r:12}
		\State Compute \emph{Virtual Commander} $C'$ (based on $C$ and $R$) and points $D'_i$\label{r:13}
		\If{$\overline{CR}=d+\mu$ \textbf{and} I am not in $D'_1$}\ Move to $D'_1$\label{r:14}
		\ElsIf{$\overline{CR}=d'$ \textbf{and} $\angle NRC>90^\circ$ \textbf{and} I am not in $D'_2$}\ Move to $D'_2$\label{r:15}
		\ElsIf{$\overline{CR}=d'$ \textbf{and} $\angle NRC<90^\circ$ \textbf{and} I am not in $D'_3$}\ Move to $D'_3$\label{r:16}
		\EndIf
	\Else\label{r:17}
		\State Compute \emph{Virtual Commander} $C'$ (based on $R$ and $N$) and points $S_i$, $S'_i$\label{r:18}
		\If{$\exists i\in\{1,2,3\}$ s.t.\ $C$ is on segment $S_iS'_i$}\label{r:19}
			\State Move to point $P$ on segment $QQ'$ such that $\overline{PQ}=\overline{CS_i}$\label{r:20}
		\EndIf
	\EndIf
\ElsIf{I am \emph{Reference}}\label{r:21}
	\If{\emph{Commander} and \emph{Number} are not tasked with moving (based on the above rules)}\label{r:22}
		\State $\gamma=$ circle centered in $C$ with radius $d$\label{r:23}
		\State $\gamma'=$ circle with diameter $CN$\label{r:24}
		\State Move to the intersection of $\gamma$ and $\gamma'$ closest to $R$\label{r:25}
	\EndIf
\EndIf
\end{algorithmic}
\end{algorithm}

Since the robots are anonymous, they first have to determine their roles, i.e., who is the Commander, etc. (line~\ref{r:1} of the algorithm). We make the assumption that there exists a disk of radius $\varepsilon$ containing only the TuringMobile (close to its center) and no other robot. Using the fact that the two closest robots must be the Commander and the Reference robot and that the two farthest robots must be the Commander and the Number robot, it is then easy to determine who is who (these properties will be preserved throughout the execution, as we will see in the next section).

Once it has determined its role, each robot executes a different branch of the algorithm (cf.~lines~\ref{r:2},~\ref{r:11}, and~\ref{r:21}). The general idea is that, when the Commander realizes that the machine is in its rest position, it decides where to move next, i.e., it chooses a final destination $D_i$. This choice is based on the number $r$ stored in the machine's ``memory'' (i.e., the number encoded by $\overline{RN}$), the relative positions of the visible robots external to the machine, and also on the application, i.e., the specific program that the TuringMobile is executing.

When the Commander has decided its final destination $D_i$, the entire machine moves by the vector $\overrightarrow{CD_i}$, and the Number robot also updates its distance from the Reference robot to represent a different real number $r'$. Again, this number is computed based on the number $r$ the machine was previously representing, the relative positions of the visible robots external to the machine, and the specific program: in general, the new distance between $N$ and $Q$ is a function $f$ of the old distance.

When this is done, the machine is in its rest position again, so the Commander chooses a new destination, and so on, indefinitely.

\paragraph*{Coordinating movements}
Note that it is not possible for all three robots to translate by $\overrightarrow{CD_i}$ at the same time, because they are non-rigid and asynchronous. If the scheduler stops them at arbitrary points during their movement, after the structure of the machine has been destroyed, they will be incapable of recovering all the information they need to resume their movement (recall that they are oblivious and they have to compute a destination point from scratch every time).

To prevent this, the robots employ various coordination techniques. First the Commander moves to the middle triangle $\tau_i$, and precisely to its base vertex $A_i$, as shown in Figure~\ref{f:synch}(a) (cf.~line~\ref{r:5} of Algorithm~\ref{a:1}). Then it positions itself on the altitude $S_iS'_i$, in such a way as to indicate the new number $r'$ that the machine is supposed to represent. That is, the Commander picks the point on $S_iS'_i$ at distance $f(\overline{NQ})$ from $S_i$ (lines~\ref{r:5:1} and~\ref{r:5:2}). Even if it is stopped by the scheduler before reaching such a point, it can recover its destination by simply drawing a ray from $A_i$ to its current position and intersecting it with $S_iS'_i$ (lines~\ref{r:6} and~\ref{r:7}).

When the Commander has reached $S_iS'_i$, it waits to let the Number robot adjust its position on the segment $QQ'$ to match that of the Commander on $S_iS'_i$, as in Figure~\ref{f:synch}(b) (lines~\ref{r:19} and~\ref{r:20}). This effectively makes the Number robot represent the new number $r'$. Note that the Number robot can do this even if it is stopped by the scheduler several times during its march, because the Commander keeps reminding it of the correct value of $r'$: since $r'$ depends on the old number $r$, the Number robot would be unable to re-compute $r'$ after it has forgotten $r$.

Once the Number robot has reached the correct position on $QQ'$, the Commander starts moving again (line~\ref{r:8}) and finally reaches $D_i$ while the other robots wait, as in Figure~\ref{f:synch}(c) (lines~\ref{r:9} and~\ref{r:10}).

When the Commander has reached $D_i$, the Number robot realizes it and makes the corresponding move (lines~\ref{r:12}--\ref{r:16}) while the other two robots wait. The destination point of the Number robot is $D'_i$, as shown in Figure~\ref{f:machine}. Finally, when the Number robot is in $D'_i$, the Reference robot realizes it and makes the final move to bring the TuringMobile back into a rest position (lines~\ref{r:21}--\ref{r:25}). Note that the number $r'$ stored in the machine is not erased after these final movements, because both the Number and Reference robot move by the same vector.

\begin{figure}[ht]
\begin{center}
\includegraphics[width=\linewidth]{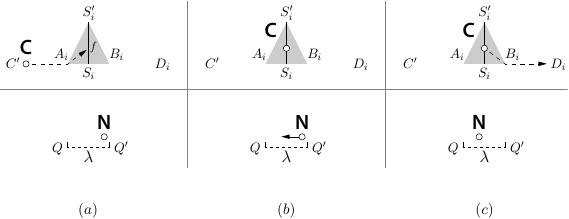}
\end{center}
\caption{Coordinated movement of the Commander and the Number robot, to cope with their asynchronous and non-rigid nature. (a) The Commander stops on $S_iS'_i$, recording the number that the machine is going to represent next (which is a function $f$ of the number currently represented by the Number robot). (b) The Number robot moves within $QQ'$ to match the Commander's position in $S_iS'_i$. (c) Finally, the Commander reaches $D_i$.}
\label{f:synch}
\end{figure}

\paragraph*{Computing the Virtual Commander}
After the Commander has left its rest position and is on its journey to $D_i$, the TuringMobile loses its initial shape, and identifying the $D_i$'s and the midway triangles becomes a non-trivial task. To simplify this task, the robots try to guess where the Commander's original rest position may have been by computing a point called the \emph{Virtual Commander} $C'$.

Assuming that the Reference and Number robots have not moved from their rest positions, the Virtual Commander is easily computed: draw a line $\ell$ through $R$ perpendicular to $RN$; then, $C'$ is the point on $\ell$ at distance $d$ from $R$ that is closest to $C$. Once we have $C'$, we can construct the points $D_i$ with respect to $C'$ (in the same way as we did in Figure~\ref{f:machine} with respect to $C$). This technique is used by Algorithm~\ref{a:1} at lines~\ref{r:3} and~\ref{r:18}.

In the special case where the Commander has reached its final destination $D_i$ and the Reference robot has not moved from its rest position (but perhaps the Number robot has moved), the Virtual Commander can also be computed. This situation is recognized because the distance between the Commander and the Reference robot is either maximum (i.e., $d+\mu$) or minimum (i.e., $d'=\sqrt{d^2+\mu^2-d\mu}$, by the law of cosines), as Figure~\ref{f:machine} shows. If the distance is maximum, then $C$ must coincide with $D_1$; otherwise, $C$ coincides with $D_2$ (if the angle $\angle NRC$ is obtuse) or $D_3$ (if the angle $\angle NRC$ is acute). Since we know the position of $R$ and one of the $D_i$'s, it is then easy to determine the other $D_i$'s. This technique is used at line~\ref{r:13}.

\paragraph*{The Reference robot's behavior}
To know when it has to start moving, the Reference robot simply executes Algorithm~\ref{a:1} from the perspective of the Commander and the Number robot: if neither of them is supposed to move, then the Reference robot starts moving (line~\ref{r:22}).

We have seen that the Number robot can determine its destination $D'_i$ solely by looking at the positions of $C$ and $R$, which remain fixed as it moves. For the Reference robot the destination point is not as easy to determine, because the distance between $C$ and $N$ varies depending on what number is stored in the TuringMobile.

However, the Reference robot knows that its move must put the TuringMobile in a rest position. The condition for this to happen is that its destination point be at distance $d$ from $C$ (line~\ref{r:23}) and form a right angle with $C$ and $N$ (line~\ref{r:24}). There are exactly two such points in the plane, but one of them has distance much greater than $\mu$ from $R$, and hence the Reference robot will pick the other (line~\ref{r:25}).

As the Reference robot moves toward such a point, all the above conditions must be preserved, due to the asynchronous and non-rigid nature of the robots. This is not a trivial requirement, and we will prove that it is indeed fulfilled in Section~\ref{s4}.

\subsection{Complete Implementation}\label{s3:2}
We have shown how to implement a basic component of the TuringMobile in $\mathbb R^2$ consisting of three robots: a Commander, a Number, and a Reference. Te basic component is able to rigidly move by a fixed distance $\mu$ in three fixed directions, $120^\circ$ apart from one another. It can also store and update a single real number.

\paragraph*{Planar layout}
We can obtain a full-fledged TuringMobile in $\mathbb R^2$ by putting several tiny copies of the basic component side by side as in Figure~\ref{f:vis}.

\begin{figure}[ht]
\begin{center}
\includegraphics[scale=1]{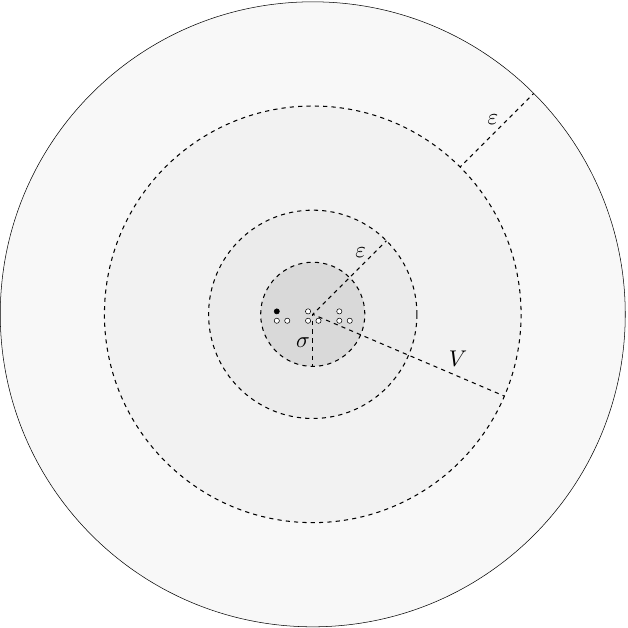}
\end{center}
\caption{Sketch of a complete TuringMobile, not drawn to scale ($\varepsilon$ and $\sigma$ should be smaller). All robots are in the central disk of radius $\sigma$; the one drawn in black is the Leader}
\label{f:vis}
\end{figure}

For the machine to work, we stipulate that there exists a disk of radius $\sigma$ that contains all the robots constituting the TuringMobile and no extraneous robot, where $\sigma\ll\varepsilon$. The distance between two consecutive basic components of the TuringMobile is roughly $s$, where $d\ll s\ll \sigma$. This makes it easy for the robots to tell the basic components apart and determine the role of each robot within its basic component.

Since a basic component of the TuringMobile is a scalene triangle, which is chiral, all its members implicitly agree on a clockwise direction even if they have different handedness. Similarly, all robots in the TuringMobile agree on a ``leftmost'' basic component, whose Commander is said to be the \emph{Leader} of the whole machine.

\paragraph*{Coordinated movements}
All the basic components of the TuringMobile are always supposed to agree on their next move and proceed in a roughly synchronous way. To achieve this, when all the basic components are in a rest position, the Leader decides the next direction among the three possible, and executes line~\ref{r:4} of Algorithm~\ref{a:1}. Then all the other Commanders see where the Leader is going, and copy its movement.

When all the Commanders are in their respective $A_i$'s, they execute line~\ref{r:5:2} of the algorithm, and so on. At any time, each robot executes a line of the algorithm only if all its homologous robots in the other basic components of the TuringMobile are ready to execute that line or have already executed it; otherwise, it waits.

When the last Reference robot has completed its movement, the machine is in a rest position again, and the coordinated execution repeats with the Leader choosing another direction, etc.

\paragraph*{Simulating a non-oblivious rigid robot}
Let a program for a rigid robot $\mathcal R$ in $\mathbb R^2$ with $k$ persistent registers and visibility radius $V$ be given. We want the TuringMobile described above to act as $\mathcal R$, even though its constituting robots are non-rigid and oblivious.

Our TuringMobile consists of $2+k$ basic components, each dedicated to memorizing and updating one real number. These $2+k$ numbers are the $x$ coordinate and the $y$ coordinate of the destination point of $\mathcal R$ and the contents of the $k$ registers of $\mathcal R$. We will call the first two numbers the \emph{$x$ variable} and the \emph{$y$ variable}, respectively.

When the TuringMobile is in a rest position, its $x$ and $y$ variables represent the coordinates of the destination point of $\mathcal R$ relative to the Leader of the machine. Whenever the TuringMobile moves by $\mu$ in some direction, these values are updated by subtracting the components of an appropriate vector of length $\mu$ from them. Of course, this update is computed by the Commanders of the first two basic components of the machine, which communicate it to their respective Number robots, as explained in Section~\ref{s3:1}.

Let $P$ be the destination point of $\mathcal R$. Since the TuringMobile can only move by vectors of length $\mu$ in three possible directions, it may be unable to reach $P$ exactly. So, the Leader always plans the next move trying to reduce its distance from $P$ until this distance is at most $2\sigma$ (this is possible because $\mu\ll d\ll \sigma$).

When the Leader is close enough to $P$, it ``pretends'' to be in $P$, and the TuringMobile executes the program of $\mathcal R$ to compute the next destination point. Recall that the visibility radius of $\mathcal R$ is $V$, and that of the robots of the TuringMobile is $V+\varepsilon$. Since $\sigma\ll \varepsilon$, each member of the TuringMobile can therefore see everything that would be visible to $\mathcal R$ if it were in $P$, and compute the output of the program of $\mathcal R$ independently of the other members. The only thing it should do when it executes the program of $\mathcal R$ is subtract the values of the $x$ and $y$ variables to everything it sees in its snapshot, discard whatever has distance greater than $V$ from the center of the snapshot, and of course discard the robots of the TuringMobile and replace them with a single robot in the center of the snapshot (representing the robot itself). Then, the robots that are responsible for updating the $x$ and $y$ variables add the relative coordinates of the new destination point of $\mathcal R$ to these variables. Similarly, the robots responsible for updating the $k$ registers of $\mathcal R$ do so.

Note that the above simulation works also in the special case where $\mathcal R$ is supposed to update its registers without moving. The Leader will move by $\mu$ in any direction, followed by the entire machine (because this is the only way the TuringMobile can update its registers), and the $x$ and $y$ variables will be updated with the old position of the Leader.

\paragraph*{Restrictions}
The above TuringMobile correctly simulates $\mathcal R$ under certain conditions. The first one is that, if all robots are indistinguishable, then no robot extraneous to the TuringMobile may get too close to it (say, within a distance of $\sigma$ of any of its members). This kind of restriction cannot be dispensed with: whatever strategy a team of oblivious robots employs to simulate a single non-oblivious robot's behavior is bound to fail if extraneous robots join the team creating ambiguities between its members. Nevertheless, the restriction can be removed if we stipulate that the members of a TuringMobile are distinguishable from all other robots.

Another difficulty comes from the fact that, if the TuringMobile is made of more than one basic component and its Commanders are all in their respective $A_i$'s and ready to update the values represented by the machine, they may get their snapshots at different times, due to asynchrony. If the environment moves in the meantime, the snapshots they get are different, and this may cause the machine to compute an incorrect destination point or put inconsistent values in its simulated registers.

There are several possible solutions to this problem: we will only mention two trivial ones. We could for instance assume the Commanders to be \emph{synchronous}, that is, make the scheduler activate them in such a way that all of them take their snapshots at the same time. This way, all Commanders get compatible snapshots and compute consistent outputs. Another possible solution is to make the TuringMobile operate in an environment where everything else is static, i.e., no moving entities are present other than the TuringMobile's members.

We stress that these restrictions make sense if a perfect simulation of $\mathcal R$ is saught. As we will see in Section~\ref{s5}, there are several other applications of the TuringMobile technique where no such restrictions are required.

\paragraph*{Higher dimensions}
Let us now generalize the above construction of a planar TuringMobile to $\mathbb R^m$, for any $m\geq 2$. We start with the same TuringMobile $\mathcal M$ with $2+k$ basic components laid out on a plane $\gamma\subset\mathbb R^m$. Since $\mathcal M$ has only two basic components for the $x$ and $y$ variables, we will add $m-2$ basic components to it, positioned as follows.

Let vectors $v_1$ and $v_2$ be two orthonormal generators of $\gamma$, and let us complete $\{v_1, v_2\}$ to an orthonormal basis $\{v_1, v_2, \dots, v_m\}$ of $\mathbb R^m$. Now, for all $i\in\{3,4,\dots,m\}$, we make a copy of the basic component of $\mathcal M$ containing the Leader, we translate it by $s\cdot v_i$, and we add it to the TuringMobile ($s$ is the same value used in the construction of the planar TuringMobile at the beginning of Section~\ref{s3:2}). Note that the Leader of this new TuringMobile $\mathcal M'$ is still easy to identify, as well as the plane $\gamma$ when $\mathcal M'$ is at rest.

Clearly, $m$ basic components allow the machine to record a destination point in $\mathbb R^m$, as opposed to $\mathbb R^2$. Additionally, the positions of the basic components with respect to $\gamma$ give the machine an $m$-dimensional sense of direction.

For instance, say that $m=3$, $\gamma$ is a horizontal plane, and $v_3$ points upward. Then, when the Leader decides to move up, it moves by $\mu$ in the direction of the basic component of the TuringMobile not lying on $\gamma$ (first stopping in a midway triangle, as per Algorithm~\ref{a:1}). The rest of $\mathcal M'$ can reconstruct the direction of $v_3$, for instance by inspecting the relative positions of the Reference robots, and move as required when the time comes. In the subsequent moves, the Leader still retains a consistent notion of up and down, and can therefore lead $\mathcal M'$ close enough to the destination point.

The same restrictions that apply to the planar TuringMobile as a simulator of course extend to its higher-dimensional versions. The next section will be devoted to proving the following theorem, which summarizes the results obtained so far.

\begin{theorem}\label{t:2}
Under the aforementioned restrictions, a rigid robot in $\mathbb R^m$ with $k$ persistent registers and visibility radius $V$ can be simulated by a team of $3m+3k$ non-rigid oblivious robots in $\mathbb R^m$ with visibility radius $V+\varepsilon$.
\end{theorem}

\section{Correctness}\label{s4}
This section is devoted to the proof of Theorem~\ref{t:2}. The crux of the proof is the following lemma, which states that a single basic component of the TurnigMobile, as described in Section~\ref{s3:1}, works as intended.

\paragraph*{The fundamental lemma}

\begin{lemma}\label{l:3}
Let a TuringMobile in $\mathbb R^2$ consisting of a single basic component execute Algorithm~\ref{a:1}, and assume that throughout the execution no object extraneous to the machine approaches any of its members by less than $\sigma$. If at some point in time $t$ the TuringMobile is in a rest position and none of its members is moving anywhere, then, at a point in time $t'>t$, the TuringMobile is in a rest position again, its Commander and Reference robot have translated by a vector of length $\mu$ in one of three predefined directions (as in Figure~\ref{f:machine}), its Number robot has correctly updated its distance from the Reference robot (according to some function $f$ of the previous distance and the TuringMobile's surrounding environment as observed by the Commander in a single snapshot taken between times $t$ and $t'$), and no member of the TuringMobile is moving anywhere.
\end{lemma}

A robot is ``not moving anywhere'' at time $t$ if it has already reached the last destination point that it has computed before time $t$, and it has not taken the next snapshot before time $t$ (although it may be taking the snapshot exactly at time $t$).

Note that for Lemma~\ref{l:3} we do not make all the restrictions of Theorem~\ref{t:2}, because we do not have to synchronize several basic components of the machine.

If Lemma~\ref{l:3} holds, then a TuringMobile in $\mathbb R^2$ with only one basic component correctly performs a single step of the execution, rigidly moving by $\mu$ and updating the real number that it is storing. By repeatedly applying this lemma, we have the correctness of the entire execution of a basic component.

Proving Theorem~\ref{t:2} is then a simple matter, because the coordination of several basic components in $\mathbb R^m$ for $m\geq 2$ is done as described in Section~\ref{s3:2}, and does not pose any problems.

\paragraph*{Proof structure}

Let us prove Lemma~\ref{l:3}. The intended behavior of the machine is for the execution to go through the following five phases in chronological order:
\begin{enumerate}
\item The Commander moves to $S_iS'_i$ (lines~\ref{r:4}--\ref{r:7} of Algorithm~\ref{a:1});
\item The Number robot moves within $QQ'$ (lines~\ref{r:19} and~\ref{r:20});
\item The Commander moves to $D_i$ (lines~\ref{r:8}--\ref{r:10});
\item The Number robot moves to $D'_i$ (lines~\ref{r:12}--\ref{r:16});
\item The Reference robot moves, bringing the machine in a rest position (lines~\ref{r:22}--\ref{r:25}).
\end{enumerate}
During each phase, only one robot is supposed to move, while the other two wait. If we can ensure this behavior, then Lemma~\ref{l:3} follows.

Recall that the robots constituting the TuringMobile are asynchronous and non-rigid. This means that we have to guarantee two things for each of the above phases:
\begin{itemize}
\item If a robot moves as per phase~$i$ and another robot sees it at any time before it has finished (due to asynchrony), the second robot does not mistakingly think that the current phase is not $i$, and hence it waits.
\item If a robot moves as per phase~$i$ and the scheduler stops it before it has reached its destination (due to non-rigidity), the robot takes another snapshot, and correctly resumes phase~$i$.
\end{itemize}

\paragraph*{Phase 1}
If the assumptions of Lemma~\ref{l:3} are satisfied, the first robot to take a snapshot after time $t$ (or exactly at time $t$) will se a TuringMobile in a rest position. As a consequence, the Virtual Commander coincides with the Commander, and therefore only the Commander is allowed to move toward some $A_i$.

While the Commander moves, its distance from the Reference robot never gets as small as $d'$ or as large as $d+\mu$ (cf.~Figure~\ref{f:machine}), hence the conditions of line~\ref{r:12} of Algorithm~\ref{a:1} are never satisfied. Also, the Virtual Commander computed with respect to $R$ and $N$ always coincides with the starting position $C'$ of the Commander, which means that the Commander will be seen on the segment $C'A_i$, implying that only the Commander will be allowed to move.

Since the Commander approaches $A_i$ by at least $\delta$ at every movement (cf.~Section~\ref{s2:2}), it eventually reaches it. When it reaches $A_i$, it chooses a destination point on $S_iS'_i$ based on a single snapshot of the environment (as required by Lemma~\ref{l:3}): once a destination point has been chosen, it never changes even if the Commander is stopped before reaching it, due to lines~\ref{r:6} and~\ref{r:7}. Since the Number robot and the Reference robot have not moved yet, the number stored in the machine is still the same as it was a time $t$, and therefore the point on $S_iS'_i$ chosen by the Commander is correctly computed by applying function $f$ to $\overline{QN}$. Again, only the Commander is allowed to move, and it eventually reaches $S_iS'_i$, for the same reasons as before.

\paragraph*{Phase 2}
When the Commander is on $S_iS'_i$, it waits until the Number robot has a distance from $Q$ of $\overline{CS_i}$. Observe that, as the Number robot moves within $QQ'$, the slope of the line $RN$ does not change, and therefore the Virtual Commander $C'$ computed with respect to $R$ and $N$ is always the position that $C$ occupied at time $t$. So, the point $S_i$ is always the same, and the Number robot keeps consitently moving toward the same destination point on $QQ'$.

As for phase~1, $\overline{CR}$ never becomes as small as $d'$ or as large as $d+\mu$, and therefore the Number robot always executes line~\ref{r:19} until it reaches the correct point on $QQ'$.

\paragraph*{Phase 3}
When $\overline{NQ}=\overline{CS_i}$, the Commander knows it has to start moving again, first to $B_i$ and then to $D_i$, due to lines~\ref{r:8}--\ref{r:10}. Again, while this happens the Virtual Commander computed based on $R$ and $N$ is always the same point $C'$, so the positions of $S_i$, $B_i$, $D_i$, etc. remain consistent, and the distance between $R$ and $C$ never gets as small as $d'$ or as large as $d+\mu$ until the Commander has reached $D_i$. In particular the Number robot never sees the Commander on $S_iS'_i$ after it has left it, and so it does not move. Eventually, the Commander reaches $D_i$.

\paragraph*{Phase 4}
When the Commander reaches $D_i$, its distance from $R$ finally becomes $d+\mu$ (if $i=1$) or $d'$ (if $i=2$ or $i=3$), and so the Number robot executes lines~\ref{r:13}--\ref{r:16} and starts moving. While the Number robot moves, the Commander does not: indeed, as long as the Number robot is tasked with moving, the Reference robot never moves (cf.~line~\ref{r:22}), and hence $\overline{CR}$ remains the same. Therefore, if the Commander computes a Virtual Commander $\widetilde C$ based on $N$ and $R$, and then computes the points $D_i$ and the midway triangles $\tau_i$ with respect to $\widetilde C$, it will never believe to be in $\widetilde C$ or in the interior of the segment $\widetilde CD_i$ or in $\tau_i$, no matter where $N$ is. This is because all such points have distance greater than $d'$ and smaller than $d+\mu$ from $R$ (cf.~Figure~\ref{f:machine}). So, the conditions of lines~\ref{r:5}--\ref{r:10} are never satisfied.

Suppose that the Commander is in $D_1$. This configuration is correctly identified by the Number robot no matter how it moves, because it is the only one in which $\overline{CR}=d+\mu$. So the Number robot computes the point $D'_1$ correctly (it does so only based on $C$ and $R$) and keeps moving toward $D'_1$ until it reaches it (line~\ref{r:14}).

Suppose now that the Commander is in $D_2$. The Number robot recognizes this configuration because $\overline{CR}=d+\mu$ and $\angle NRC > 90^\circ$. Again, it correctly computes $D'_2$ and moves toward it (line~\ref{r:15}). As the Number robot moves, the angle $\angle NRC$ grows (cf.~Figure~\ref{f:machine}), and so the condition of line~\ref{r:15} keeps being satisfied. Eventually, the Number robot reaches $D'_2$.

Finally, suppose that the Commander robot is in $D_3$. Now $\overline{CR}=d+\mu$ and $\angle NRC < 90^\circ$, and so the Number robot starts moving toward $D'_3$ (cf.~Figure~\ref{f:machine}). We have to prove that, if it stops on its way to $D'_3$ and gets a new snapshot, the inequality $\angle NRC < 90^\circ$ keeps being satisfied, and so the Number robot re-computes $D'_3$ as its destination point, until it reaches it. This is not trivial, since the angle $\angle NRC$ grows as $N$ approaches $D'_3$.

\begin{figure}[ht]
\begin{center}
\includegraphics[width=\linewidth]{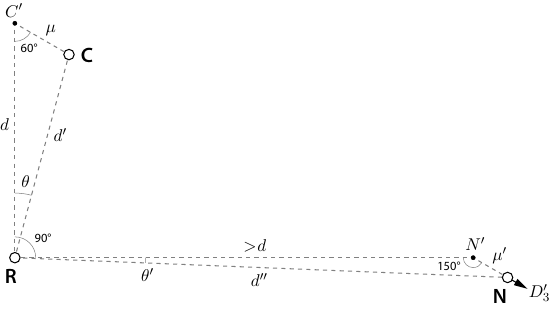}
\end{center}
\caption{As $N$ moves toward $D'_3$, we have $\theta > \theta'$, and hence $\angle NRC<90^\circ$}
\label{f:angle}
\end{figure}

The situation is illustrated in Figure~\ref{f:angle}, where $C'$ represents the starting position of the Commander and $N'$ the starting position of the Number robot. Since $\angle N'RC'=90^\circ$, proving that $\angle NRC<90^\circ$ is equivalent to proving that $\theta > \theta'$ (where $\theta'=\angle NRN'$).

By the law of sines applied to triangle $RCC'$,
\begin{equation*}
\frac{d'}{\sin 60^\circ} = \frac{\mu}{\sin \theta},
\end{equation*}
implying that
\begin{equation}\label{e:1}
\sin\theta = \frac{\sqrt 3 \mu}{2d'}.
\end{equation}
Again for the law of sines applied to triangle $RNN'$,
\begin{equation*}
\frac{d''}{\sin 150^\circ} = \frac{\mu'}{\sin \theta'},
\end{equation*}
where $\mu'=\overline{N'N}$. Hence
\begin{equation}\label{e:2}
\sin\theta' = \frac{\mu'}{2d''}.
\end{equation}

Observe that $\mu>\mu'$, because $N$ lies on $N'D'_3$, and therefore $\sqrt 3\mu > \mu'$ (i.e., the numerator of~(\ref{e:1}) is greater than that of~(\ref{e:2})). Recall that $\mu\ll d$, and so $d'<d$. Moreover, since $\overline{RN'}\geq 2d-\lambda/2$ (cf.~Figure~\ref{f:machine}) and $\lambda\ll d$, it follows that $d<\overline{RN'}$. Also observe that $\overline{RN'}<d''$, from which we obtain that $2d'<2d''$ (i.e., the denominator of~(\ref{e:1}) is smaller than that of~(\ref{e:2})). As a consequence, $\sin \theta > \sin \theta'$. Since $\mu\ll d$, both $\theta$ and $\theta'$ are acute, which means that $\theta > \theta'$ (the function $\sin x$ increases monotonically when $x\in [0, 90^\circ]$).

\paragraph*{Phase 5}
Since in the previous phases either the Commander or the Number robot was always tasked with moving, the condition of line~\ref{r:22} was never satisfied, and hence the Reference robot never moved. Now that the Commander is in $D_i$ and the Number robot is in $D'_i$, they are no longer tasked with moving, and so the Reference robot executes lines~\ref{r:23}--\ref{r:25}.

Ideally, the Reference robot should complete the translation of the TuringMobile in order to put it in a rest position again. This is achieved by moving by vector $\overrightarrow{C''C}$, where $C''$ is the initial position of the Commander (i.e., its position when phase~1 starts). Instead of trying to reconstruct $C''$, the Reference robot constructs two circumferences $\gamma$ and $\gamma'$ and moves to their nearest intersection point. Note that $\gamma'$ passes through the center of $\gamma$, and hence it has at most two intersection points with it. At least one intersection point exists: this is the point $P=R'+\overrightarrow{C''C}$, where $R'$ is the initial position of the Reference robot (which coincides with its position when phase~5 starts). If there is another intersection point $P'$ between the two circles, it must be symmetric to $P$ with respect to line $CN$, because such line passes through the centers of both circles (recall that the segment $CN$ is a diameter of $\gamma'$). So, assuming that the Commander and the Number robot do not move in this phase, $P$ remains the destination point of the Reference robot as long as the robot never crosses the line $CN$. But this is impossible, since the segment $R'P$ has length $\mu\ll d$, and therefore it cannot cross the line $CN$, whose distance from $R'$ is roughly $2d/\sqrt{5}\gg\mu$.

It remains to prove that, as the Reference robot moves toward $P$, the Commander and the Number robot remain still. Recall that, when phase~5 starts, either $\overline{CR} = d+\mu>d$ or $\overline{CR} = d'<d$. As $R$ approaches $P$ (and $C$ does not move), $\overline{CR}$ converges monotonically to $d$. It follows that $\overline{CR}$ never becomes $d+\mu$ or $d'$ again, and so the condition of line~\ref{r:12} is never satisfied.

Consider now the Virtual Commander $C'$ computed with respect to $R$ and $N$ when $R$ is strictly between $R'$ and $P$, and construct the three segments $C'D_i$ and the three midway triangles $\tau_i$ around $C'$. If we can prove that, no matter where $R$ is located in the interior of the segment $RN$, $C$ never lies on any of these segments and triangles, we are finished: indeed, this would mean that the conditions of lines~\ref{r:5}--\ref{r:10} and line~\ref{r:19} are never satisfied.

Suppose that the Number robot has moved to $D'_1$ during phase~4, which means that at the beginning of phase~5 we have $\overline{CR}=d+\mu$: this case is illustrated in Figure~\ref{f:proof1}. We have to show that $C$ does not lie on any of the solid gray lines and triangles around the Virtual Commander $C'$. It is obvious that the lines $C''C$ and $RC'$ are not parallel and intersect each other at $R$. Also, $C'$ and $N$ are always on opposite sides of $RC'$. This already implies that $C$ cannot be on the segments $C'D_1$ and $C'D_2$ or on their respective midway triangles.

\begin{figure}[h!]
\begin{center}
\includegraphics[width=\linewidth]{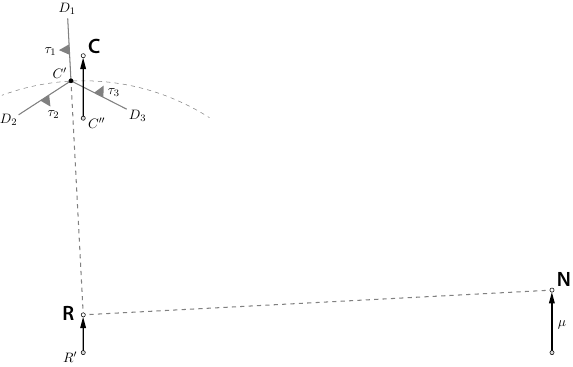}
\end{center}
\caption{Correctness of phase~5 when the Number robot has moved to $D'_1$}
\label{f:proof1}
\end{figure}

To show that $C$ does not lie on $C'D_3$ or $\tau_3$, consider the circle through $C'$ centered at $R$. Note that $C$ is always outside the circle, because its radius is $d$, but $d<\overline{RC}<d+\mu$.  Since $\mu$ can be arbitrarily small compared to $d$, the angle between $C''C$ and $RC'$ can be made arbitrarily small, as well (cf.~Figure~\ref{f:proof1}). If we take a small-enough $\mu$, the segment $C'D_3$ is entirely contained in the circle, and hence it cannot contain $C$. Moreover, since $\tau_3$ has height $\lambda\ll\mu$, by taking a small-enough $\lambda$ we ensure that $\tau_3$ is contained in the circle, too.

Suppose now that the Number robot has moved to $D'_2$ during phase~4: then, at the beginning of phase~5, $\overline{CR}=d'$ and $\angle NRC>90^\circ$. On the other hand, when $R$ reaches $P$, we have $\overline{CR}=d$ and $\angle NRC=90^\circ$. It follows that, when $R$ is strictly between $R'$ and $P$, $d'<\overline{CR}<d$ and $\angle NRC>90^\circ$ (because both quantities change monotonically), as Figure~\ref{f:proof2} shows. Since $\angle NRC>90^\circ$, $C$ cannot be located on $C'D_1$ or $C'D_3$ or $\tau_3$, because all their points $X$ satisfy $\angle NRX\leq 90^\circ$. Also, all the points in $\tau_1$ have distance at least $d+\mu/2-\lambda/2>d$ from $R$ (recall that $\lambda\ll\mu$), and so $C$ cannot lie in $\tau_1$, because $\overline{CR}<d$.

\begin{figure}[h!]
\begin{center}
\includegraphics[width=\linewidth]{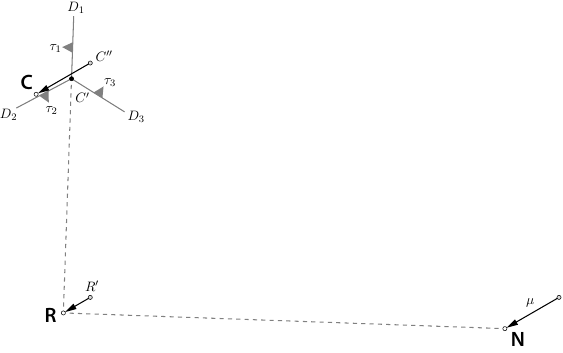}
\end{center}
\caption{Correctness of phase~5 when the Number robot has moved to $D'_2$}
\label{f:proof2}
\end{figure}

Let us show that $C$ does not lie on $C'D_2$ or $\tau_2$. Observe that $\overline{R'C''} = \overline{RC'} = d$, that $\angle RC'D_2 = \angle R'C''C = 60^\circ$, and that $R'R$ and $C''C$ are parallel (cf.~Figure~\ref{f:proof2}). It follows that the line $C''C$ is obtained by rotating line $C'D_2$ about $R$ by some angle $\theta>0$. These two lines are not parallel, and hence they intersect in a single point $K$. As $R$ approaches $P$, $\theta$ tends monotonically to $0$, and $K$ approaches the foot $U$ of the altitude from $R$ to the line $C''C$. So, if we take $\mu$ small enough with respect to $d$, we can keep $K$ as close as we want to $U$. The distance between $U$ and $C''$ is obviously minimum when $R=R'$, in which case $\overline{C''U}=d/2$. It follows that, for small-enough values of $\mu$, $\overline{C''K}$ is always as close as we want to $d/2$. Hence we have $\overline{C''K}>\mu = \overline{C''C}$, proving that $C\neq K$, and so $C$ cannot be on $C'D_2$. Also, since $\angle NRC>90^\circ$, $C$ and $\tau_2$ are always on opposite sides of $C'D_2$ (cf.~Figure~\ref{f:proof2}), and so $C$ cannot be in $\tau_2$.

Lastly, suppose that the Number robot has moved to $D'_3$ during phase~4: then, at the beginning of phase~5, $\overline{CR}=d'$ and $\angle NRC<90^\circ$, as shown in Figure~\ref{f:proof3}. Similarly to the previous case, we can prove that $C$ cannot lie on $C'D_3$ or $\tau_3$ because the line $C''C$ is obtained by rotating line $C'D_3$ about $R$ by some angle $\theta>0$ that can be made arbitrarily small by just decreasing $\mu$. Again, this implies that the intersection point between the lines $C''C$ and $C'D_3$ can be kept at a distance from $C''$ arbitrarily close to $d/2$, and can therefore never coincide with $C$, which is only $\mu$ away from $C''$. Also, because $\angle NRC<90^\circ$ ($\angle NRC$ increases monotonically and converges to $90^\circ$ as $R$ converges to $P$), $C$ and $\tau_3$ are always on opposite sides of $C'D_3$, and so $C$ cannot be in $\tau_3$.

\begin{figure}[h!]
\begin{center}
\includegraphics[width=\linewidth]{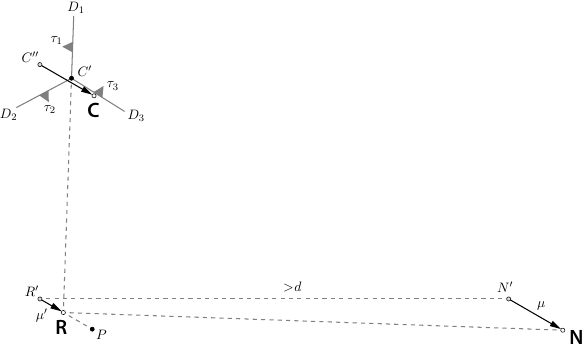}
\end{center}
\caption{Correctness of phase~5 when the Number robot has moved to $D'_3$}
\label{f:proof3}
\end{figure}

To conclude the proof, it suffices to show that $C$ and $R'$ lie on strictly opposite sides of line $RC'$: indeed, this would imply that $C$ is not on the segments $C'D_1$ and $C'D_2$ or in their respective midway triangles, because these lie on the same side of $RC'$ as $R'$ or on the line $RC'$ itself (cf.~Figure~\ref{f:proof3}). To prove this claim, consider a Cartesian coordinate system with origin in $R$ and $x$ axis oriented as $\overrightarrow{R'N'}$. Let $0<\mu'=\overline{R'R}<\mu$. Since the line $RR'$ forms an angle of $60^\circ$ with the $y$ axis, the coordinates of $R'$ are
\begin{equation*}
R'=\left(-\frac{\sqrt 3 \mu'}{2},\ \frac{\mu'}{2}\right).
\end{equation*}
We therefore have 
\begin{equation*}
C''=R'+(0,\,d)=\left(-\frac{\sqrt 3 \mu'}{2},\ \frac{\mu'}{2}+d\right)
\end{equation*}
and
\begin{equation*}
C=C''+\left(\frac{\sqrt 3 \mu}{2},\ -\frac{\mu}{2}\right) = \left(\frac{\sqrt 3 (\mu-\mu')}{2},\ d-\frac{\mu-\mu'}{2}\right).
\end{equation*}
We also have
\begin{equation*}
N'=R'+\left(\overline{R'N'},\,0\right) = \left(\overline{R'N'}-\frac{\sqrt 3 \mu'}{2},\ \frac{\mu'}{2}\right)
\end{equation*}
and
\begin{equation*}
N=N'+\left(\frac{\sqrt 3 \mu}{2},\ -\frac{\mu}{2}\right) = \left(\overline{R'N'}+\frac{\sqrt 3 (\mu-\mu')}{2},\ \frac{\mu'-\mu}{2}\right).
\end{equation*}
It follows that the line $RN$ has equation
\begin{equation*}
y = \frac{\mu'-\mu}{2\cdot \overline{R'N'}+\sqrt 3 (\mu-\mu')}\ x.
\end{equation*}
Since the line $RC'$ is orthogonal to $RN$, it has equation
\begin{equation}\label{e:3}
y = \left(\frac{2\cdot \overline{R'N'}}{\mu-\mu'}+\sqrt 3\right) x.
\end{equation}
Observe that $RC'$ passes through the origin and its slope is positive. Hence $R'$ lies above this line, because its $x$ coordinate is negative and its $y$ coordinate is positive.

Let us now plug the $x$ coordinate of $C$ in~(\ref{e:3}):
\begin{equation}\label{e:4}
y = \left(\frac{2\cdot \overline{R'N'}}{\mu-\mu'}+\sqrt 3\right)\cdot \frac{\sqrt 3 (\mu-\mu')}{2} = \sqrt 3 \cdot \overline{R'N'} + \frac{3(\mu-\mu')}{2}.
\end{equation}
Recall from the discussion on phase~4 that $\overline{R'N'}>d$ (it corresponds to $\overline{RN'}$ in Figure~\ref{f:angle}), and therefore the $y$ in~(\ref{e:4}) is abundantly greater than $d$. On the other hand, the $y$ coordinate of $C$ is $d-(\mu-\mu')/2$, which is smaller than $d$, implying that $C$ lies below the line $RC'$. We conclude that $C$ and $R'$ lie on opposite sides of $RC'$.

We have just proved that the Reference robot keeps moving until it reaches $P$, thus bringing the TuringMobile in a rest position again, say at time $t'$. We ultimately observe that the real number stored in the machine at time $t'$ is the same the one the Commander computed in phase~1 and that the Number robot copied during phase~2. This is because the Number robot and the Reference robot, during phases~4 and~5 respectively, have moved by $\mu$ in the same direction: so, at the end of phase~5, they have the same distance they had at the end of phase~2.

\section{Applications}\label{s5}
In this section we discuss some applications of the TuringMobile. We also prove that the basic TuringMobile constructed in Section~\ref{s3:1} is minimal, in the sense that no smaller team of oblivious robots can accomplish the same tasks.

\subsection{Exploring the Plane}
The first elementary task a basic TuringMobile in $\mathbb R^2$ can fulfill is that of \emph{exploring} the plane. The task consists in making all the robots in the TuringMobile see every point in the plane in the course of an infinite execution. We first assume that the three members of the TuringMobile are the only robots in the plane. Later in this section, we will extend our technique to other types of scenarios and more complex tasks.

\begin{theorem}\label{t:4}
A basic TuringMobile consisting of three robots in $\mathbb R^2$ can explore the plane.
\end{theorem}
\begin{proof}
Recall that a basic TuringMobile can store a single real number $r$ and update it at every move as a result of executing a real RAM program with input $r$. In particular, the TuringMobile can count how many times it has moved by simply starting its execution with $r=0$ and computing $r:=r+1$ at each move.

Moreover, the Commander chooses the direction of the next move (in the form of a point $D_i$, see Figure~\ref{f:machine}) by executing another real RAM program with input $r$. If $r$ is an integer, the Commander can therefore compute any Turing-computable function on $r$, and so it can decide to move to $D_1$ the first time, then to $D_2$ twice, then to $D_3$ three times, to $D_1$ four times, and so on. This pattern of moves is illustrated in Figure~\ref{f:explore}, and of course it results in the exploration of the plane, because the visibility radius of the robots is much greater than the step $\mu$.
\end{proof}

\begin{figure}[ht]
\begin{center}
\includegraphics[width=\linewidth]{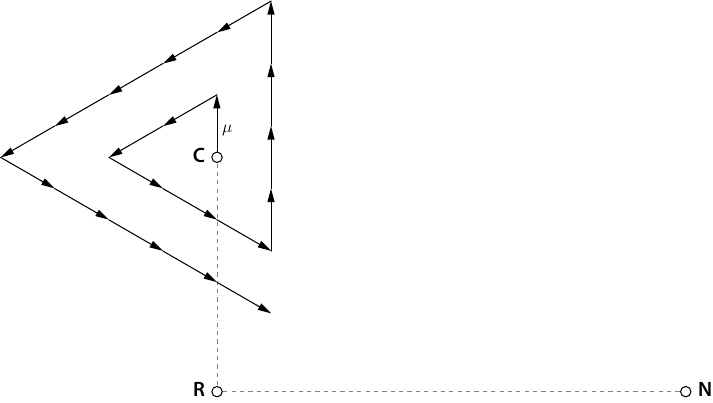}
\end{center}
\caption{Exploration of the plane by a basic TuringMobile}
\label{f:explore}
\end{figure}

\subsection{Minimality of the Basic TuringMobile}
We can use the previous result to prove indirectly that our basic TuringMobile design is minimal, because no team of fewer than three oblivious robots in $\mathbb R^2$ can explore the plane.

\begin{theorem}\label{t:5}
If only one or two oblivious identical robots with limited visibility are present in $\mathbb R^2$, they cannot explore the plane, even if the scheduler lets them move synchronously and rigidly.
\end{theorem}
\begin{proof}
Assume that a single oblivious robot is given in $\mathbb R^2$ (hence no other entities or obstacles are present). Since the robot always gets the same snapshot, it always computes the same destination point in its local coordinate system, and so it always translates by the same vector. As a consequence, it just moves along a straight ray, and therefore it cannot explore the plane.

Let two oblivious robots be given, and suppose that their local coordinate systems are oriented symmetrically. Whether the robots see each other or not, if they take their snapshots simultaneously, they get identical views, and so they compute destination points that are symmetric with respect to $O$. If they keep moving synchronously and rigidly, $O$ remains their midpoint. So, if the robots have visibility radius $V$, they see each other if and only if they are in the circle $\gamma$ of radius $V/2$ centered in $O$.

Let $O$ be the midpoint of the robots' locations, and consider a Cartesian coordinate system with origin $O$. Without loss of generality, when the robots do not see each other, they move by vectors $(1,0)$ and $(-1,0)$, respectively. Let $\xi$ be the half-plane $y\geq V$, and observe that $\xi$ lies completely outside $\gamma$.

It is obvious that the robots cannot explore the entire plane if neither of them ever stops in $\xi$. The first time one of them stops in $\xi$, it takes a snapshot from there, and starts moving parallel to the $x$ axis, thus never seeing the other robot again, and never leaving $\xi$. Of course, following a straight line through $\xi$ is not enough to explore all of it.
\end{proof}

\subsection{Near-Gathering with Limited Visibility}
The exploration technique can be applied to several more complex problems. The first we describe is the \emph{Near-Gathering} problem, in which all robots in the plane must get in the same disk of a given radius $\varepsilon$ (without colliding) and remain there forever. It does not matter if the robots keep moving, as long as there is a disk of radius $\varepsilon$ that contains them all.

It is clear that solving this problem from every initial configuration is not possible, and hence some restrictive assumptions have to be made. The usual assumption is that the initial visibility graph of the robots be connected~\cite{xxFlPSW05,xPaPV15}. Here we make a different assumption: there are three robots that form a basic TuringMobile somewhere in the plane, and each robot not in the TuringMobile has distance at least $\varepsilon$ from all other robots. (Actually we could weaken this assumption much more, but this simple example is good enough to showcase our technique.) Also, in the existing literature on the Near-Gathering problem it is always assumed that the robots agree on at least one coordinate axis, but here we do not need this assumption.

Say that all robots in the plane have a visibility radius of $V\gg\varepsilon$, and that the TuringMobile moves by $\mu\ll\varepsilon$ at each step. The TuringMobile starts exploring the plane as above, and it stops in a rest position as soon as it finds a robot whose distance from the Commander is smaller than $V/2$ and greater than $\varepsilon$. On the other hand, if a robot is not part of the TuringMobile, it waits until it sees a TuringMobile in a rest position at distance smaller than $V/2$. When it does, it moves to a designated area $\mathcal A$ in the proximity of the Commander. Such an area has distance at least $3d$ from the Commander, so no confusion can arise in the identification of the members of the TuringMobile. If several robots are eligible to move to $\mathcal A$, only one at a time does so: note that the layout of the TuringMobile itself gives an implicit total order to the robots around it. Observe that the robots cannot form a second TuringMobile while they move to $\mathcal A$: in order to do so, two of them would have to move to $\mathcal A$ at the same time and get close enough to a third robot. Once they enter $\mathcal A$, the robots position themselves on a segment much shorter than $d$, so they cannot possibly be mistaken for a TuringMobile.

Once the eligible robots have positioned themselves in $\mathcal A$, the TuringMobile resumes its exploration of the plane, and the robots in $\mathcal A$ copy all its movements. Of course, at each step the TuringMobile waits for all the robots in $\mathcal A$ to catch up before carrying on with the exploration. Now, if the total number of robots in the plane is known, the TuringMobile can stop as soon as all of them have joined it. Otherwise, the machine simply keeps exploring the plane forever, eventually collecting all robots. In both cases, the Near-Gathering problem is solved.

\subsection{Pattern Formation with Limited Visibility}
Suppose $n$ robots are tasked with forming a given \emph{pattern} consisting of a multiset of $n$ points: this is the \emph{Pattern Formation} problem, which becomes the \emph{Gathering} problem in the special case in which the points are all coincident. For this problem, it does not matter where the pattern is formed, nor does its orientation or scale.

Again, the Pattern Formation problem is unsolvable from some initial configurations, so we make the same assumptions as with the Near-Gathering problem. The algorithm starts by solving the Near-Gathering problem as before. The only difference is that now there is a second tiny area $\mathcal B$, attached to $\mathcal A$ (and still far enough from the TuringMobile), which the robots avoid when they join $\mathcal A$. This is because this second area will later be used to form the pattern.

Since $n$ is known, the TuringMobile knows when it has to interrupt the exploration of the plane because all robots have already been found. At this point, the robots switch algorithms: one by one, they move to $\mathcal B$ and form the pattern. This task is made possible by the presence of the TuringMobile, which gives an implicit order to all robots, and also unambiguously defines an embedding of the pattern in $\mathcal B$. So, each robot is implicitly assigned one point in $\mathcal B$, and it moves there when its turn comes.

If $n=3$ or $n=4$, there are uninteresting ad-hoc algorithms to do this: so, let us assume that $n\geq 5$. The first to move are the robots in $\mathcal A$: this part is easy, because they all lie on a small segment, which already gives them a total order, and allows them to move one by one. The robots only have to be careful enough not to collide with other robots before reaching their final positions. Again, this is trivial, because only one robot is allowed to move at a time.

When this part is done, there are at least two robots in $\mathcal B$, all of which have distance much smaller than $d$ from each other. Then the members of the TuringMobile join $\mathcal B$ as well, in order from the closest to the farthest. Each of them chooses a position in $\mathcal B$ based on the robots already there and the remnants of the TuringMobile. Moreover, the members of the TuringMobile that have not started moving to $\mathcal B$ yet cannot be mistaken for robots in $\mathcal B$, because they are at a greater distance from all others (and vice versa).

Note that, when the last robot leaves the TuringMobile and joins $\mathcal B$, it is able to find its final location because there are already at least four robots there, which provide a reference frame for the pattern to be formed. When this last robot has taken position in $\mathcal B$, the pattern is formed.

\subsection{Higher Dimensions}
Everything we said in this section pertained to robots in the plane. However, we can generalize all our results to robots in $\mathbb R^m$, for $m\geq 2$. Recall that, at the end of Section~\ref{s3:2}, we have described a TuringMobile for robots in $\mathbb R^m$, which can move within a specific plane $\gamma$ exactly as a bidimensional TuringMobile, but can also move back and forth by $\mu$ in all other directions orthogonal to $\gamma$.

Now, extending our results to $\mathbb R^m$ actually boils down to exploring the space with a TuringMobile: once we can do this, we can easily adapt our techniques for the Near-Gathering and the Pattern Formation problem, with negligible changes.

There are several ways a TuringMobile can explore $\mathbb R^m$: we will only give an example. Consider the exploration of the plane described at the beginning of this section, and let $P_i$ be the point reached by the Commander after its $i$th move along the spiral-like path depicted in Figure~\ref{f:explore} ($P_0$ is the initial position of the Commander).

Our $m$-dimensional TuringMobile starts exploring $\gamma$ as if it were $\mathbb R^2$. Whenever it visits a $P_i$ for the first time, it goes back to $P_0$. From $P_0$, it keeps making moves orthogonal to $\gamma$ until it has seen all points in $\mathbb R^m$ whose projection on $\gamma$ is $P_0$ and whose distance from $P_0$ is at most $i$. Then it goes back to $P_0$, moves to $P_1$, and repeats the same pattern of moves in the section of $\mathbb R^m$ whose projection on $\gamma$ is $P_1$. It then does the same thing with $P_2$, etc. When it reaches $P_{i+1}$ (for the first time), it goes back to $P_0$, and proceeds in the same fashion. By doing so, it explores the entire space $\mathbb R^m$.

Note that this algorithm only requires the TuringMobile to count how many moves it has made since the beginning of the execution: thus, the machine only has to memorize a single integer. The direction of the next move according to the above pattern is then obviously Turing-computable given the move counter.

\section{Conclusions}\label{s6}
We have introduced the TuringMobile as a special configuration of oblivious non-rigid robots that can simulate a rigid robot with memory. We have also applied the TuringMobile to some typical robot problems in the context of limited visibility, showing that the assumption of connectedness of the initial visibility graph can be dropped if a unique TuringMobile is present in the system. Our results hold not only in the plane, but also in Euclidean spaces of higher dimensions.

The simplest version of the TuringMobile (Section~\ref{s3:1}) consists of only three robots, and is the smallest possible configuration with these characteristics (Theorems~\ref{t:4} and~\ref{t:5}).
Our generalized TuringMobile (Section~\ref{s3:2}), which works in $\mathbb R^m$ and simulates $k$ registers of memory, consists of $3m+3k$ robots (Theorem~\ref{t:2}). We believe we can decrease this number to $m+k+3$ by putting all the Number robots in the same basic component and adopting a more complicated technique to move them. However, minimizing the number of robots in a general TuringMobile is left as an open problem.

Our basic TuringMobile design works if the robots have the same radius of visibility, because that allows them to implicitly agree on a unit of distance. We could remove this assumption and let each of them have a different visibility radius, but we would have to add a fourth robot to the TuringMobile for it to work (as well as keep the TuringMobile small compared to \emph{all} these radii).

Recall that, in order to encode and decode arbitrary real numbers we used the $\alpha$ function and its inverse, which in turn are computed using the $\arctan$ and the $\tan$ functions. However, using transcendental functions is not essential: we could achieve a similar result by using only comparisons and arithmetic operations. The only downside would be that such a real RAM program would not run in a constant number of machine steps, but in a number of steps proportional to the value of the number to encode or decode. With this technique, we would be able to dispense with the trigonometric functions altogether, and have our robots use only arithmetic operations and square roots to compute their destination points.

\end{document}